\documentclass[journal]{IEEEtran}
\usepackage{graphicx}
\usepackage{booktabs}
\usepackage{latexsym,bm,amsmath,amssymb}
\usepackage{epsfig}
\usepackage{graphicx}
\usepackage{subfigure}
\usepackage{multirow}
\usepackage{algorithm}
\usepackage{algorithmic}
\usepackage{color}
\usepackage{mathrsfs}
\usepackage{array}
\usepackage{enumerate}
\usepackage{xcolor}
\usepackage[misc]{ifsym}
\usepackage[nosort]{cite}
\usepackage{picins}
\usepackage{threeparttable}
\usepackage{textcomp}
\usepackage{caption}

\usepackage{amsthm}

\newtheorem{lemma}{Lemma}
\newtheorem{remark}{Remark}

\ifCLASSINFOpdf
\else
\fi
\begin{document}
%
\title{Optimal Energy-Delay in Energy Harvesting Wireless Sensor Networks with Interference Channel}
%
%
%
\author{Dongbin Jiao,~
        Liangjun Ke,~
        Shengbo Liu,~and~Felix T.S. Chan
\thanks{D. Jiao and L. Ke
are with the School of Electronic and Information Engineering and State Key Laboratory for Manufacturing Systems Engineering, Xi'an Jiaotong University, Xi'an, 710049, China
(e-mail: dbjiao@stu.xjtu.edu.cn; keljxjtu@xjtu.edu.cn).}
\thanks{S. Liu is with the School of Information Science and Engineering, Key Laboratory of Underwater Acoustic Communication and Marine Information Technology Ministry of Education, Xiamen University, Xiamen, 361005, China (e-mail: liushb@stu.xmu.edu.cn).}
\thanks{F. Chan is with the Department of Industrial and Systems Engineering, The Hong Kong Polytechnic University, Hung Hom, Hong Kong, China (e-mail: f.chan@polyu.edu.hk).}
}

\maketitle
\begin{abstract}
In this work, we investigate the capacity allocation problem in the energy harvesting wireless sensor networks (WSNs) with interference channel. For the fixed topologies of data and energy, we formulate the optimization problem when the data flow remains constant on all data links and each sensor node harvests energy only once in a time slot. We focus on the optimal data rates, power allocations and energy transfers between sensor nodes in a time slot. Our goal is to minimize the total delay in the network under two scenarios, i.e., no energy transfer and energy transfer. Furthermore, since the optimization problem is non-convex and difficult to solve directly. By considering the network with relatively high Signal-to-Interference-plus-Noise Ratio (SINR), the non-convex optimization problem can be transformed into a convex optimization problem by convex approximation. We attain the properties of optimal solution by Lagrange duality and solve the convex optimization problem by CVX solver. The experimental results demonstrate that the total delay of the energy harvesting WSNs with interference channel is more than that in the orthogonal channel; and the energy transfer can help to decrease the total delay. Moreover, we also discuss the extension of our work.
\end{abstract}


%
\IEEEpeerreviewmaketitle

\section{Introduction}\label{intro}
\IEEEPARstart{E}{nergy}  harvesting is a promising solution to provide self-sustain ability  and extend the lifetime for energy-limit wireless sensor networks (WSNs) \cite{ulukus2015energy}. Thus it has attracted much attention from researchers in recent years. However, energy harvesting process from the natural environment is instable, due to the time change of the day, the season or other factors. Wireless energy transfer (WET) as a friendly means of compensating energy, can transfer energy from some energy-rich sensor nodes to  others with energy-hungry so as to enhance the overall network performance \cite{8269105}. Meanwhile, due to the broadcast nature of wireless communications, the data signals of simultaneous transmissions can not avoid to interfere with each other in the same frequency band. As a result, it decreases the network performance.

Because of these considerations, we investigate the energy harvesting WSNs and concentrate on the delay minimization problem of the WSNs with interference channel. The delay of every data link is determined by the information rate on the link, which is monotonically decreased as the rate of the link for the fixed data flow over it \cite{bertsekas1992data}. The information rate is monotonically increasing in SINR. We focus on the capacity assignment problem which is similar to reference \cite{bertsekas1992data}. In particular, compared with the special case, in which information and energy transfer channels are orthogonal to each other \cite{gurakan2016optimal}, we consider the general case of communication model. In other words, the data transmission channels are interfered with each other. This is a more realistic and meaningful model of the capacity assignment problem.

Therefore, by considering the energy consumption and power allocation for the fixed data flow, we formulate the capacity assignment problem in the energy harvesting WSNs with interference channel as a non-convex optimization problem, which is constrained by data flow conservation conditions, information rate requirements, energy and power consumption. Employing the relatively high SINR, the non-convex optimization problem can be transformed into a convex optimization problem by convex approximation in "log-sum-exp" form \cite{boyd2004convex}. The solution properties of transformed capacity allocation problem is derived by Lagrange duality. Then it is available to search the optimal Lagrange multiplier and obtain the optimal solution to minimize total delay for the energy harvesting WSNs with interference channel in a time slot. Finally, we solve the approximate convex problem by CVX solver \cite{grant2008cvx}.

Our study is related to and based on the previous classical works on capacity allocation problem in communication networks \cite{bertsekas1992data}. In \cite{xiao2004simultaneous}, the simultaneous routing and resource allocation (SRRA) is investigated. A capacitated multi-commodity flow model is used to describe the data flows in the wireless networks. The optimization problem is solved by the dual-decomposition method. A general flow-based analytical framework is presented in \cite{xi2008node}. In order to balance aggregate user utility, total network cost, power control, rate allocation, routing, and congestion control are jointly optimized in wireless networks. However, the previous classical works have not considered the energy harvesting and energy cooperation. Reference \cite{fouladgar2013information} investigates the optimization problem of simultaneous information and energy flows in graph-based communication networks with energy transfer. Though references \cite{gurakan2016optimal} and \cite{xu2017optimization}  study the optimization problem of the joint information transmission and energy transfer, they neglect the interference among the data flow signals. These motivate us to consider a general capacity assignment problem which is to minimize total delay in the energy harvesting WSNs with interference channel.

It is worth noting that although we utilize a similar mathematical approach to that in \cite{gurakan2016optimal} for modeling and solving the capacity assignment problem, our study is significantly different from the previous studies: the previous studies only consider a special case where the data transmission channels are orthogonal to each other, rather than consider the impact of data transmission interference. However, the more realistic case is that data transmission channels are interfered with each other, which is one of the critical issues to be tackled in this study. Therefore, we need to remodel the capacity assignment problem for the energy harvesting WSNs with interference channel  in a time slot.

In this paper, our main contributions are as follows:
\begin{itemize}
  \item We investigate a general and meaningful model of capacity assignment problem in the energy harvesting WSNs with interference channel.
  \item Considering relatively high SINR, we transform the non-convex optimization problem into a convex one by convex approximation, and also derive the optimal solution properties by Lagrange duality.
  \item Numerical results show that the interference signals significantly affect the network performance; the energy transfer can help to decrease the total network delay.
\end{itemize}

The rest of this paper is structured as follows. Section \ref{system-model} introduces the network model and problem formulation. Section \ref{CAP-single-time-slot} investigates capacity assignment problem with interference channel in a time slot. Section \ref{performance results} demonstrates the performance results. Finally, Section \ref{conclusion} concludes the paper.

\section{System Model and Problem Formulation} \label{system-model}
In this study, each sensor node not only has the capability of harvesting energy and sensing data from the ambient environment, but it also can transmit or receive energy and data. As the data transmission channels are interfered with each other, the interference signals among the data flows may be unavoidable. Hence, we consider an energy harvesting WSNs model with interference channel.

Let $G = (V,E)$ be a directed and connectivity graph modeling $N$ sensor nodes which are placed randomly and seamlessly in a certain area. The vertices set $V$ = \{$v_0, v_1, \ldots, v_N$\} is composed of one sink node and $N$ sensor nodes. The edges set $E$ is composed of the communication links between the sensor nodes, i.e., $(v_i, v_j)\in E$, if and only if a node $v_i$ can send a message to a node $v_j$ with the power constraint $p_{ij}$.

A \emph{data collection tree} $T = (V_T,E_T)$ \cite{imon2015energy} is constructed for the energy harvesting WSNs with sink $v_0$ at level 0 as shown in Fig.  \ref{Interference-model}. It is an acyclic spanning subgraph of $G = (V,E)$ where $V_T = V$ and $E_T \subseteq E$. In the data collection tree $T$, each sensor node $v_n$ can collect the sensing data from the area of interest and then store it for future transmission in a data buffer. Each sensor node $v_n$ has to send the sensing data to sink $v_0$ periodically in multi-hop fashion and \emph{half-duplex} mode under interference channel. Sensor nodes $v_i$ and $v_j$ are siblings if they have the same parent.
Note that a sensor node can be either a transmitter, a relay or a receiver, which is determined by its location in WSNs. For brevity, the ordered pair $(v_i, v_j)$ is replaced by $(i,j)$ in the following sections. Throughout the paper, we denote sensor node indices by the first subscripts $i$, $j$ and $n$. The subscript $i$ and $j$ denote the start node and the end node at each link (i.e., data link and energy link), respectively.

\subsection{Network Data Flow Model} \label{Network-Data-Flow-Model}
Let us denote the data link $(i, j)$ as $l \in 1 \ldots L$ \footnote{The data link can be denoted $(i,j)$ or $l$, they can be interchangeable in this paper.}. The topology of data flows can be described by an $N\times L$ matrix $\mathbf{A}$. The entries of matrix $\mathbf{A}$ can be defined by $a_{nl}$, which is incident with sensor node $n$ and data link $l$. More precisely, each entry $a_{nl}$ is defined as
\begin{equation}
    a_{nl} =
   \begin{cases}
   1,  &\mbox{if $n = i$} \\
   -1, &\mbox{if $n = j$} \\
   0,  &\mbox{otherwise}.
   \end{cases}
\end{equation}
Let us define $\mathcal{I}_d(n)$ as the set of incoming data links to sensor node $v_n$ and $\mathcal{O}_d(n)$ as the set of outgoing data links from sensor node $v_n$, respectively. Assume that the data flow $d_l$ on each data link follows the uniform distribution $U(0,1]$. The set of data flows $\{d_{l} | l \in E_T \}$ is referred to as the $L$-dimensional flow vector. The divergence vector $\mathbf{s}$ associated with the data flow vector $\mathbf{d}$ is an $N$-dimensional vector which indicates the nonnegative amount of outside data flow injected into the sensor node $v_n$. Suppose that the data flow is lossless over links. For every sensor node $v_n$, the flow conservation conditions can be expressed as
\begin{equation}
s_n = \sum_{{l}\in \mathcal{O}_d(n)}d_{l} - \sum_{{l}\in \mathcal{I}_d(n)}d_{l}, \ \ \forall n \in V_T.
\end{equation}
The data flows conservation through the total WSNs can be rewritten as
\begin{equation}
\mathbf{Ad^T = s}.
\end{equation}
Moreover, the data flow $d_{l}$  over each data link $l$ can't exceed the information carrying capacity $c_{l}$, i.e.,
\begin{equation}
d_{l} \leq c_{l}, \ \ \forall l \in E_T.
\end{equation}
\subsection{Network Energy Flow Model} \label{Network-Energy-Flow-Model}
In this section, we present the energy model for the case where each sensor node has a single energy harvest in a time slot.
\subsubsection{Energy Harvesting Model}
Each sensor node powered can harvest energy from the ambient environment. Since the transmission consumption is the most significant amount of energy, we only account for energy consumption of transmitting data in this study. It is assumed that the energy harvesting sensor node has a capacity battery $B_{max}$ which is large enough. The capacity of storage is considered to be constant, i.e., energy outage and circuitry cost are negligible. Since energy harvesting sources are with random nature, the energy arrivals are considered as an independent and identically distributed ($i.i.d.$) Poisson distribution $\mathbf{P}(\lambda)$ with parameter $\lambda$ \cite{huang2013spatial,adu2018energy}. We assume that the energy arrivals occur only once in a time slot. Let $E_n$ denote the harvested energy of a sensor node $v_n$ in a time slot, $E_n\in (0, B_{max}]$. The harvested energy in a time slot can be exploited only in a later time slot.

\subsubsection{Energy Cooperation Model} \label{energy-cooperation-model}
Energy cooperation depends upon the statistics of the energy harvesting and the energy consumption of the sensor nodes. In general, for a sensor node $v_n$, the more data flow is transmitted, the more energy is required. In order to replenish energy of energy-hungry sensor nodes, the technique of wireless energy cooperation \cite{gurakan2013energy} is adopted in our study. It is assumed that the energy is \emph{ unidirectionally} transferred from the sensor node $v_i$ to the sensor node $v_j$ in a time slot, the transfer efficiency is $\eta_{ij}$, $\eta_{ij} \in (0,1]$, due to energy loss in transmission and conversion.
\subsubsection{Energy Flow Model}
In the previous analysis, we utilize $N$-dimensional vector $\mathbf{E}$ to present the harvested energy vector for the WSNs. In energy transfer process, the wireless energy links are similar to data links. The wireless energy link $q$ is also denoted as an ordered pair $(i,j)$ in energy routing. The energy can be sent from the sensor node $v_i$ to the sensor node $v_j$ over energy link $q$, $q \in 1 \ldots Q$, if the energy of the sensor node $v_j$ is not enough energy to operate. The energy transfer efficiency is $\eta_q$ on each energy link $q$ where $\eta_q \in (0,1]$. It implies that $\delta_i$ amount of energy is transferred on wireless energy link $q$ from the sensor node $v_i$ to the sensor node $v_j$; and the sensor node $v_j$ receives $\eta_q \delta_i$ amount of energy. The request of energy transfer is known in advance whereas the amount of transferred energy is unknown. The topology of energy flow can be denoted by an $N\times Q$ matrix $\mathbf{B}$. The entries of the matrix $\mathbf{B}$ can be defined by $b_{nq}$, which is incident with sensor node $n$ and wireless energy link $q$. More specifically, each entry $b_{nq}$ can be described as
\begin{equation}
    b_{nq} =
   \begin{cases}
   1,  &\mbox{if $n = i$} \\
   -\eta, &\mbox{if $n = j$} \\
   0,  &\mbox{otherwise}.
   \end{cases}
\end{equation}
We define $\mathcal{O}_q(n)$ and $\mathcal{I}_q(n)$ as the set of outgoing and incoming wireless energy links at the sensor node $v_n$, respectively. The variable $x_{q}$ is the amount of energy transferred. Let vector $\mathbf{x}$ be the $L$-dimensional energy flow vector.

\subsection{Communication Model} \label{Communication-Model}
For the energy harvesting WSNs with interference channel, we focus on minimizing the total delay and enhancing the network performance in order to ensure that sensing data on each data link can reach the sink as quickly as possible. It is similar to \cite{bertsekas1992data,gurakan2016optimal}, we assume that each time slot is large enough and the delay on the data link $l$ follows the $M/M/1$ queueing model in this work. It can be defined as
\begin{equation} \label{data-delay}
D_{l} = \frac{d_{l}}{c_{l} - d_{l}},
\end{equation}
where $d_{l}$ is the amount of data flow and $c_{l}$ is the information carrying capacity of communication link $l$ in which $d_{l} \leq c_{l}, \forall l \in E_T$.

In this study, we consider a tree-based energy harvesting WSNs with interference channel.
As shown in Fig.  \ref{Interference-model}, there are only 5 active links at the first time slot since we employ \emph{half-duplex} sensor nodes. Meanwhile, the network has 5 energy cooperation links,  which can transfer energy to sensor nodes required in order to guarantee that the sensing data can be successfully sent to the receivers at the time slot. In Fig. \ref{Interference-model}, we assume that the active link $l_8$ is the primary link, the receiver $v_3$ not only receives the data flow signal from the transmitter $v_8$, but also receives the interference signals from other transmitters $v_1$, $v_9$, $v_{12}$ and $v_{13}$. The interference signals are represented by red dashed lines with arrows. Meanwhile, the sensor node $v_7$ can transfer energy to the sensor node $v_8$ through the energy link $q_{14}$.
At the same time, other receivers also receive interference signals from active links transmitters except themselves. For brevity, we do not label them in Fig. \ref{Interference-model}.
\begin{figure}[htb]
\begin{center}
$\begin{array}{l}
\includegraphics[width=3.0in]{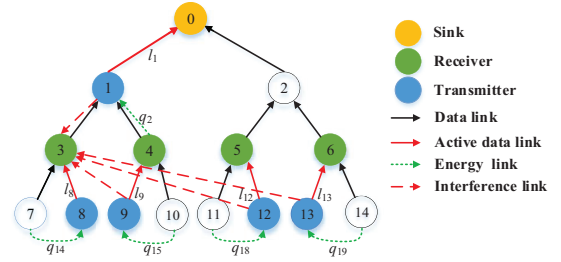}\\
 \end{array}$
\end{center}
\caption{Interference channel model of data flows with half-duplex mode.} \label{Interference-model}
\end{figure}
Hence, the data flow signals generate link interference to each other. The wireless interference signals degrade the information rate of data links and lead to greater delay in the network.

The baseband complex channel coefficient which remains constant from sensor node $v_i$ to sensor node $v_j$ is denoted by $h_{ij}$. The channel gain matrix $\mathbf{G}$ is defined by $G_{ij} = \|h_{ij}\|^2$,
which is dependent on various factors such as path loss, shadowing and fading effects. The diagonal entries $G_{ll}$ are gains of primary links, and the off-diagonal entries $G_{\bar{l}l} (\bar{l}\neq l)$ are interference gains among active data links.
Thus, the received SINR of data link $l$ is
\begin{equation} \label{Rij}
SINR_{l}(\mathbf{p}) = \frac{G_{ll} p_{l}}{\sum_{\bar{l} \neq l}G_{\bar{l}l} p_{\bar{l}} + \sigma_{l}},
\end{equation}
where $p_{l}$ denotes the depleted power which transmits data flow signal from the sensor node $v_i$ to the sensor node $v_j$ in a time slot, with channel grain $G_{ll}$ and channel noise power $\sigma_{l}$ \cite{fu2010fast}. For notational simplicity, we employ  $\mathbf{p} = \{p_{l}|l \in E_T \}$ as transmission power vector. In this paper, the power and energy can be interchangeable in a unit of time slot.

According to the Shannon formula, the information
carrying capacity (or information rate) $c_{l}$ of data link $l$ can be expressed as
\begin{equation} \label{shanon-fomula}
c_{l} = \frac{1}{2}\log(1 + SINR_{l}(\mathbf{p})),
\end{equation}
where all logarithms in our study are taken to the base $e$.

At every sensor node $v_n$, the total power depleted \footnote{In contrast to transmission power consumption, the energy consumption of sensing data is ignored in our study.} on transmission data link $l$ and energy link $q$ are constrained by the usable energy as:
\begin{equation}
\sum_{{l}\in \mathcal{O}_d(n)}p_{l}  \leq E_n + \sum_{{q}\in \mathcal{I}_q(n)} \eta_q x_q, \forall n\in V_T. \label{energy availability constraints}
\end{equation}
Let $\mathbf{K = A^+}$, where $(a^+)_{nl} = max\{a_{nl}, 0\}$, which only distinguish the outgoing links at each sensor node $n$. Hence, the energy availability constraints in Eq. \eqref{energy availability constraints} can be rewritten as
\begin{equation}
\mathbf{Kp + Bx \leq E}.
\end{equation}
\section{Capacity Assignment Problem in Energy Harvesting WSNs with Interference Channel} \label{CAP-single-time-slot}
We consider the capacity assignment problem in WSNs with interference channel for a single energy harvesting sensor node in a time slot. Assume that the data flow assignments $d_{l}$ on all data links are fixed and available for harvested energy and transferred energy. The total delay $D$ in a WSNs is
\begin{equation}
D = \sum_{l \in E_T}\frac{d_{l}}{c_{l} - d_{l}}.
\end{equation}
Hence the goal of minimizing total delay in the energy harvesting WSNs with interference channel can be written as
\begin{subequations}\label{total-delay}
\begin{align}
\min_{c_{l},p_{l},x_{q}}  &\sum_{l \in E_T} \frac{d_{l}}{c_{l} - d_{l}} \\
\text{s.t.}\ &\mathbf{Kp + Bx \leq E}\\
&d_{l} \leq c_{l}, \forall l \in E_T \\
&x_q \geq 0.
\end{align}
\end{subequations}

As shown in Fig. \ref{Interference-model}, because the data transmission signals of active links interfere with each other, each data flow signal can not perform interference cancelation and is treated as an additive noise compared with the primary link signal. By utilizing the information rate $c_{l}$ in Eq. \eqref{shanon-fomula}, the minimizing total delay in the energy harvesting WSNs  with interference channel is
\begin{subequations} \label{mode(DFDFinformationrate)}
\begin{align}
\min_{p_{l},x_{q}}  &\sum_{l\in E_T} \frac{d_{l}}{\frac{1}{2}\log\left(1 + \frac{G_{ll} p_{l}}{\sum_{\bar{l} \neq l}G_{\bar{l}l} p_{\bar{l}} + \sigma_{l}}\right) - d_{l}} \label{mode(DFDFinformationrate)Obj} \\
\text{s.t.}\ &\mathbf{Kp + Bx \leq E} \\
&p_{l}\geq \frac{\sum_{\bar{l} \neq l}G_{\bar{l}l} p_{\bar{l}} + \sigma_{l}}{G_{ll}} \left(e^{2d_{l}} - 1\right),\forall \ l \in E_T \label{mode(DFDFinformationrate)Con}\\
&x_q \geq 0.
\end{align}
\end{subequations}
By analysing \eqref{mode(DFDFinformationrate)}, we find that the minimizing of the total delay depends on the maximizing of the information carrying capacity $c_{l}$. Meanwhile, because the information carrying capacity $c_{l}$ is a monotonically increasing function of $SINR_{l}(\mathbf{p})$, the maximizing of information carrying capacity $c_{l}$ depends on the maximizing of the $SINR_{l}(\mathbf{p})$.

Note that the optimization problem \eqref{mode(DFDFinformationrate)} is non-convex since both the objective function \eqref{mode(DFDFinformationrate)Obj} and the constrain condition \eqref{mode(DFDFinformationrate)Con} are non-convex in terms of transmission power vector $\mathbf{p}$, and it is not straightforward to attain the optimal solution. Therefore, we need to study the fundamental properties of the optimization problem \eqref{mode(DFDFinformationrate)} and transform it into the convex optimization problem.
\subsection{Convex Approximation} \label{convex-approximation}
We can get a convex approximation for capacity assignment problem with interference channel when the SINRs are relatively high (e.g., SINRs $\geq$ 5 or 10).
The information carrying capacity (or information rate) $c_{l}$ by using the Eq. (\ref{Rij}) can be rewritten as
\begin{equation} \label{shanon-fomula-approximation}
\begin{aligned}
c_{l}(\mathbf{p}) &\approx \frac{1}{2}\log(SINR_{l}(\mathbf{p})) \\
&=\frac{1}{2} \log \left(\frac{G_{ll} p_{l}}{\sum_{\bar{l} \neq l}G_{\bar{l}l} p_{\bar{l}} + \sigma_{l}} \right)                         \\
&= - \frac{1}{2} \log \left(\frac{\sum_{\bar{l} \neq l}G_{\bar{l}l} p_{\bar{l}} + \sigma_{l}}{G_{ll} p_{l}} \right)                    \\
&= - \frac{1}{2} \log \left(\frac{\sigma_{l}  p_{l}^{-1}}{G_{ll}} + \frac{\sum_{\bar{l} \neq l}G_{\bar{l}l} p_{\bar{l}} p_{l}^{-1}}{G_{ll}}  \right).
\end{aligned}
\end{equation}
Let $\tilde{p}_{l} = \log (p_{l})$, i.e., $p_{l} = e^{\tilde{p}_{l}}$ for $l\in E_T$, we define
\begin{equation}
\begin{aligned}
\tilde{c_{l}}(\tilde{p}) &= c_{l}(\mathbf{p}(\tilde{p})) \\
& = - \frac{1}{2} \log \left(\frac{\sigma_{l}  e^{-\tilde{p}_{l}}}{G_{ll}} + \frac{\sum_{\bar{l} \neq l}G_{\bar{l}l} e^{\tilde{p}_{\bar{l}}-\tilde{p}_{l}}}{G_{ll}}  \right),
\end{aligned}
\end{equation}
where the functions $\tilde{c_{l}}(\tilde{p})$ are concave in the vector $\tilde{p}$.

With the approximation information carrying capacity formula, the optimization problem \eqref{mode(DFDFinformationrate)} can be reformulated as
\begin{subequations} \label{convex-form}
\begin{align}
&\min_{\tilde{p}_{l},x_{q}}  \sum_{l \in E_T} \frac{d_{l}}{- \frac{1}{2} \log \left(\frac{\sigma_{l}  e^{-\tilde{p}_{l}}}{G_{ll}} + \frac{\sum_{\bar{l} \neq l}G_{\bar{l}l} e^{\tilde{p}_{\bar{l}}-\tilde{p}_{l}}}{G_{ll}}  \right) - d_{l}} \label{convex-form-a}\\
\text{s.t.}\ &\mathbf{Kp + Bx \leq E} \\
&e^{\tilde{p}_{l}}\geq \frac{\sum_{\bar{l} \neq l}G_{\bar{l}l} e^{\tilde{p}_{\bar{l}}} + \sigma_{l}}{G_{ll}} e^{2d_{l}}, \forall \ l \in E_T \label{convex-form-c}\\
&x_q \geq 0,
\end{align}
\end{subequations}
where the objective function \eqref{convex-form-a} is convex function in the new variable $\tilde{p}_{l}$ \cite{boyd2004convex}. The information carrying capacity constraint \eqref{convex-form-c} is convex function in $\tilde{p}_{l}$ and $d_{l}$. This means that the optimization problem (\ref{convex-form}) is a convex optimization problem and the global optimal solution can be found.
\begin{remark}
Here we use the approximation $\frac{1}{2}\log(1 + SINR_{l}(\mathbf{p})) \approx \frac{1}{2}\log(SINR_{l}(\mathbf{p}))$ which is reasonable for the optimization problem \eqref{mode(DFDFinformationrate)}, since $\frac{1}{2}\log(SINR_{l}(\mathbf{p}))  \leq \frac{1}{2}\log(1 + SINR_{l}(\mathbf{p}))$. This implies that the approximation is an underestimate and a more tighten constraint for the information carrying capacity $c_{l}(\mathbf{p})$. Therefore, the solution of convex problem \eqref{convex-form} is always feasible to the original optimization problem \eqref{mode(DFDFinformationrate)}.
\end{remark}
\subsection{Properties of Capacity Assignment Problem with Interference Channel}
For convex optimization problem \eqref{convex-form}, we form the dual problem by introducing Lagrange multiplier $\lambda \in R^N$, $\beta \in R^L$ and $\gamma \in R^Q$. The Lagrangian function is given by
\begin{equation}\label{Lagrangian-function}
\begin{small}
\begin{aligned}
 &L(\tilde{p}_{l}, x_q, \lambda, \beta, \gamma) \\
=&\sum_{l\in E_T} \frac{d_{l}}{- \frac{1}{2} \log \left(\frac{\sigma_{l}  e^{-\tilde{p}_{l}}}{G_{ll}} + \frac{\sum_{\bar{l} \neq l}G_{\bar{l}l} e^{\tilde{p}_{\bar{l}}-\tilde{p}_{l}}}{G_{ll}}  \right) - d_{l}} \\
&+\sum_n \lambda_n \left(\sum_{{l}\in \mathcal{O}_d(n)}e^{\tilde{p}_{l}} - E_n - \sum_{{q}\in I_q(n)} \eta_q x_q \right)  \\
&- \sum_{l\in E_T} \beta_{l} \left(e^{\tilde{p}_{l}} - \frac{\sum_{\bar{l} \neq l}G_{\bar{l}l} e^{\tilde{p}_{\bar{l}}} + \sigma_{l}}{G_{ll}} e^{2d_{l}} \right) - \sum_q \gamma_q x_q.
\end{aligned}
\end{small}
\end{equation}
The Lagrangian function \eqref{Lagrangian-function} corresponds to Lagrange dual function $\overline{Q}: R^N \times R^L \times R^Q \rightarrow R$ as
\begin{equation}\label{Lagrangian-dual-function}
\begin{aligned}
\overline{Q}(\lambda, \beta, \gamma)&=  \inf_{\tilde{p}_{l},x_q} L(\tilde{p}_{l}, x_q, \lambda, \beta, \gamma).
\end{aligned}
\end{equation}
The dual optimization problem is
\begin{subequations} \label{Lagrangian-problem}
\begin{align}
&\max\  \overline{Q}(\lambda, \beta, \gamma) \\
&\text{s.t.}\ \lambda \geq 0, \beta \geq 0, \gamma \geq 0.
\end{align}
\end{subequations}
The KKT optimality conditions hold for the convex optimization problem \eqref{convex-form}, thus we have
\begin{equation} \label{KKT-p}
\begin{small}
\begin{aligned}
&\frac{\partial L}{\partial \tilde{p}_{l}} = \frac{\partial t_{l}(\tilde{p}_{l})}{\partial \tilde{p}_{l}} + e^{\tilde{p}_{l}}\left[\lambda_{i(l)} - \left(\beta_{l}-\beta_{\bar{l}} \sum_{\bar{l}\neq l}\frac{G_{l\bar{l}} e^{2d_{\bar{l}}}}{G_{\bar{l}\bar{l}}}\right)\right] = 0, \\
&\forall l,\bar{l} \in E_T
\end{aligned}
\end{small}
\end{equation}
\begin{equation} \label{KKT-x}
\frac{\partial L}{\partial x_q} = - \eta_q \lambda_{j(q)} - \gamma_q = 0, \forall i,j \in V_T, \ \forall q,
\end{equation}
where
\begin{equation}
\begin{aligned}
&t_{l}(\tilde{p}_{l}) \triangleq d_{l}\left[- \frac{1}{2} \log \left(\frac{\sigma_{l}  e^{-\tilde{p}_{l}}}{G_{ll}} + \frac{\sum_{\bar{l} \neq l}G_{\bar{l}l} e^{\tilde{p}_{\bar{l}}-\tilde{p}_{l}}}{G_{ll}}  \right) - d_{l}\right]^{-1}.
\end{aligned}
\end{equation}
The complementary slackness conditions are
\begin{equation} \label{complementaryslackness-lambda}
\begin{aligned}
&\lambda_n \left(\sum_{{l}\in \mathcal{O}_d(n)}e^{\tilde{p}_{l}} - E_n - \sum_{{q}\in I_q(n)} \eta_q x_q \right) = 0,
\forall n \in V_T
\end{aligned}
\end{equation}
\begin{equation} \label{complementaryslackness-beta}
\beta_{l} \left(e^{\tilde{p}_{l}} - \frac{\sum_{\bar{l} \neq l}G_{\bar{l}l} e^{\tilde{p}_{\bar{l}}} + \sigma_{l}}{G_{ll}} e^{2d_{l}}\right) = 0, \forall l \in E_T
\end{equation}
\begin{equation} \label{complementaryslackness-gamma}
\gamma_q x_q = 0, \ \forall q.
\end{equation}

We extend Lemmas 1 and 2 in \cite{gurakan2016optimal} and derive some properties about the optimal power allocation with interference channel as follows.
\begin{lemma} \label{lemma-beta-zero}
The feasibility of the convex optimization problem \eqref{convex-form}  requires $\beta_{l} = 0, \forall l \in E_T$.
\end{lemma}

\begin{proof}
The proof is similar procedure in \cite{gurakan2016optimal}. If the convex optimization problem \eqref{convex-form} is feasible, the objective function \eqref{convex-form-a} must be guaranteed to bound. The constraint condition \eqref{convex-form-c} for any data link $l$ means that the objective function \eqref{convex-form-a} is unbounded. Thus the constraint condition \eqref{convex-form-c} must strictly satisfy the inequalities for all data link $l$. From Eq. \eqref{complementaryslackness-beta} we can conclude that $\beta_{l} = 0, \forall l \in E_T$.
\end{proof}

\begin{lemma} \label{leman-single-time-slot}
At each sensor node $v_n$, the optimal power allocation with interference channel among data links satisfies
\begin{equation} \label{lemma-t-double}
\frac{\partial t_{l}(\tilde{p}_{l})}{\partial \tilde{p}_{l}} = \frac{\partial t_i(\tilde{p}_i)}{\partial \tilde{p}_{i}}, \forall l \in E_T, \forall i \in \mathcal{O}_d(n).
\end{equation}
\end{lemma}

\begin{proof}
The proof is similar procedure in \cite{gurakan2016optimal}.
Combining Eq. \eqref{KKT-p} and Lemma \ref{lemma-beta-zero}, we attain
\begin{equation}
\frac{\partial t_{l}(\tilde{p}_{l})}{\partial \tilde{p}_{l}} = -e^{\tilde{p}_{l}}\lambda_{i(l)}, \forall l \in  E_T.
\end{equation}
Since the outgoing links $l$ and $i$ reside to the same sensor node $n$, we have
\begin{equation}
\frac{\partial t_{l}(\tilde{p}_{l})}{\partial \tilde{p}_{l}} = -e^{\tilde{p}_{l}}\lambda_{i} = \frac{\partial t_i(\tilde{p}_i)}{\partial \tilde{p}_{i}}.
\end{equation}
Thus we can conclude that Eq. \eqref{lemma-t-double} holds.
\end{proof}

%
In the next subsections, we separately solve the convex optimization problem \eqref{convex-form} under two cases, i.e., no energy transfer and energy transfer.

\subsection{Case without Energy Transfer}
As energy transfer does not occur in this case, we have $x_{q} = 0, \forall q$. Thus the convex optimization problem \eqref{convex-form} becomes only in respect of $\tilde{p}_{l}$ as follows:
\begin{subequations} \label{convex-form-no-energy-transfer}
\begin{align}
\min_{\tilde{p}_{l}}  &\sum_{l\in E_T} \frac{d_{l}}{- \frac{1}{2} \log \left(\frac{\sigma_{l}  e^{-\tilde{p}_{l}}}{G_{ll}} + \frac{\sum_{\bar{l} \neq l}G_{\bar{l}l} e^{\tilde{p}_{\bar{l}}-\tilde{p}_{l}}}{G_{ll}}  \right) - d_{l}} \\
\text{s.t.}\ &\sum_{{l}\in \mathcal{O}_d(n)}e^{\tilde{p}_{l}} \leq E_n, \forall n \in V_T  \\
&e^{\tilde{p}_{l}}\geq \frac{\sum_{\bar{l} \neq l}G_{\bar{l}l} e^{\tilde{p}_{\bar{l}}} + \sigma_{l}}{G_{ll}} e^{2d_{l}}, \forall \ l \in E_T.
\end{align}
\end{subequations}
Since we employ \emph{half-duplex} WSNs, the optimization problem can be considered $\bar{L}$ active data links in the energy harvesting WSNs with interference channel as
\begin{subequations} \label{delay-no-energy-transfer}
\begin{align}
\min_{\tilde{p}_{l}} &\sum^{\bar{L}}_{i=1} \sum_{l \in \mathcal{O}_d(n)} \frac{-2d_{l}}{\log \left(\frac{\sigma_{l}  e^{-\tilde{p}_{l}}+\sum_{\bar{l} \neq l}G_{\bar{l}l} e^{\tilde{p}_{\bar{l}}-\tilde{p}_{l}}}{G_{ll}} \right) - d_{l}} \\
\text{s.t.}\ &\sum_{{l}\in \mathcal{O}_d(n)}e^{\tilde{p}_{l}} \leq E_n, \forall n \in V_T  \label{delay-no-energy-transfer-b}\\
&e^{\tilde{p}_{l}}\geq \frac{\sum_{\bar{l} \neq l}G_{\bar{l}l} e^{\tilde{p}_{\bar{l}}} + \sigma_{l}}{G_{ll}} e^{2d_{l}}, \forall \ l \in E_T. \label{delay-no-energy-transfer-c}
\end{align}
\end{subequations}
If the optimization problem \eqref{delay-no-energy-transfer} is feasible, then it requires
\begin{equation}
\sum_{{l}\in \mathcal{O}_d(n)}\frac{\sum_{\bar{l} \neq l}G_{\bar{l}l} e^{\tilde{p}_{\bar{l}}} + \sigma_{l}}{G_{ll}} e^{2d_{l}} \leq E_n,
\end{equation}
which we assume that it holds. Similar to \eqref{Lagrangian-function}, \eqref{delay-no-energy-transfer} corresponding to Lagrangian function $\hat{L}$ with $\lambda \in R^N$ is
\begin{equation}\label{Lagrangian-function-on-energy}
\begin{small}
\begin{aligned}
&\hat{L}(\tilde{p}_{l},\lambda) \\
= & \sum^{\bar{L}}_{i=1}\sum_{l \in \mathcal{O}_d(n)} \frac{-2d_{l}}{\log \left(\frac{\sigma_{l}  e^{-\tilde{p}_{l}}+\sum_{\bar{l} \neq l}G_{\bar{l}l} e^{\tilde{p}_{\bar{l}}-\tilde{p}_{l}}}{G_{ll}} \right) - d_{l}} \\
&+\sum_n \lambda_n \left(\sum_{{l}\in \mathcal{O}_d(n)}e^{\tilde{p}_{l}} - E_n \right).
\end{aligned}
\end{small}
\end{equation}
Meanwhile, the KKT optimality condition is
\begin{equation} \label{Lagrangian-function-on-energy- KKT}
\frac{\partial \hat{L}}{\partial \tilde{p}_{l}} = \frac{\partial t_{l}(\tilde{p}_{l})}{\partial \tilde{p}_{l}} + e^{\tilde{p}_{l}} \lambda = 0, \forall l \in \mathcal{O}_d(n)
\end{equation}
and the complementary slackness condition is
\begin{equation} \label{complementary-no-energy}
\lambda \left(\sum_{{l}\in \mathcal{O}_d(n)}e^{\tilde{p}_{l}} - E_n \right) = 0,\ \forall l \in E_T.
\end{equation}
\begin{equation}
\begin{small}
\begin{aligned}
&\frac{\partial t_{l}(\tilde{p}_{l})}{\partial \tilde{p}_{l}}\\ =
&-\frac{1}{2}d_{l}\left[-\frac{1}{2}\log\left(\frac{\sigma_{l}  e^{-\tilde{p}_{l}}}{G_{ll}} + \frac{\sum_{\bar{l} \neq l}G_{\bar{l}l} e^{\tilde{p}_{\bar{l}}-\tilde{p}_{l}}}{G_{ll}}  \right)-d_{l}\right]^{-2} \\
&+\frac{1}{2}\sum_{\bar{l} \neq l}\left\{d_{\bar{l}}\left[-\frac{1}{2}\log\left(\frac{\sigma_{\bar{l}} e^{-\tilde{p}_{\bar{l}}}+\sum_{k \neq \bar{l}}G_{k\bar{l}} e^{\tilde{p}_{k}-\tilde{p}_{\bar{l}}}}{G_{\bar{l}\bar{l}}}  \right) \right. \right.\\
&\left. \left.-d_{\bar{l}}\right]^{-2}\left(\frac{G_{l\bar{l}}e^{\tilde{p}_{l}}} {\sigma_{\bar{l}} +\sum_{k \neq \bar{l}}G_{k\bar{l}} e^{\tilde{p}_{k}}} \right)\right\}, \forall l,\bar{l},k \in \bar{L}
\end{aligned}
\end{small}
\end{equation}
From Eq. \eqref{Lagrangian-function-on-energy- KKT}, we have
\begin{equation} \label{lambda-no-energy-transfer}
\begin{small}
\begin{aligned}
\lambda & = -\frac{\partial t_{l}(\tilde{p}_{l})}{\partial \tilde{p}_{l}}e^{-\tilde{p}_{l}} = \\
&\frac{d_{l}}{2e^{\tilde{p}_{l}}}\left[-\frac{1}{2}\log\left(\frac{\sigma_{l}  e^{-\tilde{p}_{l}}}{G_{ll}} + \frac{\sum_{\bar{l} \neq l}G_{\bar{l}l} e^{\tilde{p}_{\bar{l}}-\tilde{p}_{l}}}{G_{ll}}  \right)-d_{l}\right]^{-2}\\
&-\sum_{\bar{l} \neq l}\left\{\frac{d_{\bar{l}}}{2}\left[-\frac{1}{2}\log\left(\frac{\sigma_{\bar{l}} e^{-\tilde{p}_{\bar{l}}}+\sum_{k \neq \bar{l}}G_{k\bar{l}} e^{\tilde{p}_{k}-\tilde{p}_{\bar{l}}}}{G_{\bar{l}\bar{l}}}  \right) \right. \right.\\
&\left. \left.-d_{\bar{l}}\right]^{-2}\left(\frac{G_{l\bar{l}}} {\sigma_{\bar{l}} +\sum_{k \neq \bar{l}}G_{k\bar{l}} e^{\tilde{p}_{k}}} \right)\right\}, \forall l,\bar{l},k \in \bar{L}
\end{aligned}
\end{small}
\end{equation}
where $\bar{L}$ is the number of active data links in a time slot.

For the total energy constraint condition Eq. \eqref{delay-no-energy-transfer-b}, the optimal power allocation can be found by searching the optimal $\lambda^*$.

\begin{remark}
The constraint condition \eqref{delay-no-energy-transfer-c} is not included in the Lagrangian function (\ref{Lagrangian-function-on-energy}), since the constraint condition \eqref{delay-no-energy-transfer-c} will always hold when the convex optimization problem \eqref{delay-no-energy-transfer} is feasible.
\end{remark}
\subsection{Case with Energy Transfer}
Next, we solve the case with energy transfer, which implies $x_q \geq 0$ for some energy links $q$. The convex optimization problem \eqref{convex-form} becomes
\begin{subequations} \label{convex-form-energy-transfer}
\begin{align}
\min_{\tilde{p}_{l},x_{q}}  &\sum_{l \in E_T} \frac{d_{l}}{- \frac{1}{2} \log \left(\frac{\sigma_{l}  e^{-\tilde{p}_{l}}}{G_{ll}} + \frac{\sum_{\bar{l} \neq l}G_{\bar{l}l} e^{\tilde{p}_{\bar{l}}-\tilde{p}_{l}}}{G_{ll}}  \right) - d_{l}} \\
\text{s.t.}\ &\sum_{{l}\in \mathcal{O}_d(n)}e^{\tilde{p}_{l}} \leq E_n + \sum_{{q}\in \mathcal{I}_q(n)} \eta_q x_q, \forall n \in V_T  \\
&e^{\tilde{p}_{l}}\geq \frac{\sum_{\bar{l} \neq l}G_{\bar{l}l} e^{\tilde{p}_{\bar{l}}} + \sigma_{l}}{G_{ll}} e^{2d_{l}}, \forall \ l \in E_T \\
&x_q \geq 0.
\end{align}
\end{subequations}
According to the \emph{half-duplex} mode, the optimization problem \eqref{convex-form-energy-transfer} which has $\bar{L}$ active data links in the energy harvesting WSNs with interference channel can be written as
\begin{subequations} \label{convex-form-energy-transfer-active}
\begin{align}
\min_{\tilde{p}_{l},x_{q}}  &\sum^{\bar{L}}_{i=1} \sum_{l \in \mathcal{O}_d(n)} \frac{d_{l}}{- \frac{1}{2} \log \left(\frac{\sigma_{l}  e^{-\tilde{p}_{l}}}{G_{ll}} + \frac{\sum_{\bar{l} \neq l}G_{\bar{l}l} e^{\tilde{p}_{\bar{l}}-\tilde{p}_{l}}}{G_{ll}}  \right) - d_{l}} \\
\text{s.t.}\ &\sum_{{l}\in \mathcal{O}_d(n)}e^{\tilde{p}_{l}} \leq E_n + \sum_{{q}\in \mathcal{I}_q(n)} \eta_q x_q, \forall n \in V_T  \\
&e^{\tilde{p}_{l}}\geq \frac{\sum_{\bar{l} \neq l}G_{\bar{l}l} e^{\tilde{p}_{\bar{l}}} + \sigma_{l}}{G_{ll}} e^{2d_{l}}, \forall \ l \in E_T \\
&x_q \geq 0.
\end{align}
\end{subequations}

As in Section \ref{energy-cooperation-model}, it is assumed that some energy $x_q > 0$ is transferred from the sensor node $v_i$ to the sensor node $v_j$ over energy link $q$. Since sensor node $v_i$ only transfers energy and does not transmit data, the energy causality constraint condition on sensor node $v_j$ is denoted as
\begin{equation} \label{optimal-lambda-power-allocation}
\sum_{{l}\in \mathcal{O}_d(j)}e^{\tilde{p}_{l}}(\lambda_j^*) = E_j + \eta_q x_q.
\end{equation}
Therefore, by combining Eq. \eqref{lambda-no-energy-transfer} and Eq. \eqref{optimal-lambda-power-allocation}, we can attain optimal power allocations if we find the optimal $\lambda_j^*$.
The Lagrangian method can provide some ideas and in-depth insight on the above-defined optimization problem. However, it is difficult to find a close-form optimal solution. Therefore, We use the CVX solver \cite{grant2008cvx} to tackle the optimization problems \eqref{delay-no-energy-transfer} and \eqref{convex-form-energy-transfer-active} in this paper.

\section{Experiment Results and Analysis} \label{performance results}
We provide some simple experimental results to demonstrate the results of the optimal energy-delay polices in the energy harvesting WSNs with interference channel. Note that we only consider the total delay of all active links in the network in a time slot, thus the power and energy can be interchangeable. We conduct our experiment on a PC with the Intel(R) Core(TM) i7-7700, 3.60 GHz CPU, 8GB RAM and Windows 8 (version 6.2). We use CVX 2.1 \cite{grant2008cvx} which is implemented in MATLAB 9.2 (version R2017a) to solve the optimization problems.

\subsection{Simulation Results}
In the simulations, a tree-based WSNs topologies are considered.
\begin{figure}[htb]
\begin{center}
$\begin{array}{l}
\includegraphics[width=3.0in]{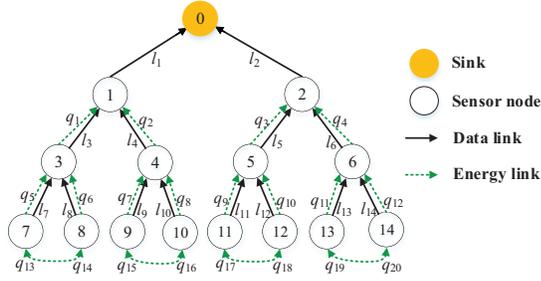}\\
 \end{array}$
\end{center}
\caption{Data and energy topologies.} \label{routing model}
\end{figure}
Fig. \ref{routing model} shows the data and energy topologies in energy harvesting WSNs, which has 1 sink (i.e., $v_0$), 14 sensor nodes, 14 directed data links and 20 directed energy links. It is noted that each leaf sensor node only needs to transfer energy from its sibling neighboring sensor node; each parent sensor node needs to transfer energy from children sensor nodes in order to transmit successfully heavy sensing data from itself and children sensor nodes; and the sink node does not need to transfer energy since it is not energy-limited. Meanwhile, the \emph{half-duplex} mode is adopted in the network system. In other words, there are only few active links in a time slot. In Fig. \ref{Interference-model}, we observe that there are 5 active links keeping simultaneous communication at the first time slot.

At each time slot, the energy arrivals follow an $i.i.d$ Poisson distribution $\mathbf{P}(\lambda)$ with $\lambda = 8$, and the data flow on each data link follows the uniform distribution $U(0,1]$. Similar to reference \cite{johansson2003simultaneous}, all the receivers have the same noise power $\sigma_{ij} = 1\times 10^{-5}$ units; all diagonal entries of the channel grain matrix $\mathbf{G}$ are set to 1 and the off-diagonal entries are attained by the uniform distribution $U(0,0.01]$. Energy transfer efficiency $\eta_q$ is set to 0.6 on all energy links \cite{watfa2011multi}.

As an example, we adopt the data and energy  topologies in Fig. \ref{Interference-model} to perform evaluation the optimization problem.
The fixed data flows are $\mathbf{d} = [d_{l_1},d_{l_8},d_{l_9},d_{l_{12}},d_{l_{13}}]^T  = [0.4585,0.8752,0.6869,0.2313,0.4887]^T$ units. The energy arrival vector $\mathbf{E_1} = [9,10,7,8,9]^T$ units and $\mathbf{E_2} = [11,10,8,4,6]^T$ units denote transmitters $\{v_1,v_8,v_9,v_{12}, v_{13}\}$ and transferring energy sensor nodes $\{v_4,v_7,v_{10},v_{11}, v_{14}\}$, respectively. The energy transfer efficiency vector is $\bm{\eta}$  = $[0.6,0.6,0.6,0.6,0.6]^T$ \footnote{Here we only give data flow of active links, corresponding to the energy of sensor nodes and the efficiency of energy transfer.}. The solution results of optimization problem under two scenarios (i.e., no energy transfer and energy transfer) are shown in the right half of Table \ref{Solution results}. In order to further confirm the significance of our study, we also perform the optimization problem of orthogonal channel \cite{gurakan2016optimal} in the tree-based network topologies. The solution results are shown in the left half of Table \ref{Solution results}.
\begin{table*}[htbp]
\begin{threeparttable}
\caption{Solution results of optimization problem under both orthogonal channel and interference channel at the first time slot.}\label{Solution results}
\begin{center}
\setlength{\tabcolsep}{0.80em}
\begin{tabular}{ |c|c|c|c|c|c|c|c|c|c|c|c|c| }
  \hline
  \multirow{3}*{\textbf{Link}} &\multicolumn{5}{c|}{\textbf{Orthogonal channel}} &\multicolumn{7}{c|}{\textbf{Interference channel}}\\ \cline{2-13}
  &\multicolumn{2}{c|}{\textbf{No energy transfer}} & \multicolumn{3}{c|}{\textbf{Energy transfer}} & \multicolumn{3}{c|}{\textbf{No energy transfer}} & \multicolumn{4}{c|}{\textbf{Energy transfer}} \\ \cline{2-13}
  \multirow{3}*{} & \textbf{Power} &\textbf{Delay} & \textbf{Power} & \textbf{TE} &\textbf{Delay} & \textbf{Power} & \textbf{SINR} &\textbf{Delay} & \textbf{Power} & \textbf{TE} &\textbf{SINR} &\textbf{Delay} \\
 \hline
  $l_1$ & 8.8143    & \multirow{5}*{0.3740} & 15.6000 & 11.0000 & \multirow{5}*{0.3622} & 5.1660  & 78.6533 & \multirow{5}*{1.8858} & 8.2649  & 7.9520  & 78.6532  &\multirow{5}*{1.8857}\\
  $l_8$ & 10.0000   & \multirow{5}*{}       & 16.0000 & 10.0000 &\multirow{5}*{}        & 10.0000 &143.1230 & \multirow{5}*{}       & 16.0000 & 10.0000 & 143.1436 &\multirow{5}*{}\\
  $l_9$ & 7.0000    & \multirow{5}*{}       & 11.8000 & 8.0000  & \multirow{5}*{}       & 4.6663  &57.5294  & \multirow{5}*{}       & 7.4654  & 6.2319  & 57.5311  &\multirow{5}*{}\\
  $l_{12}$ & 6.4475 & \multirow{5}*{}       & 10.4000 & 4.0000  & \multirow{5}*{}       & 2.5360  &14.3840  & \multirow{5}*{}       & 4.0573  & 1.2875  & 14.3839  &\multirow{5}*{}\\
  $l_{13}$ & 9.0000 & \multirow{5}*{}       & 12.6000 & 6.0000  & \multirow{5}*{}       & 3.5185  &43.8209  & \multirow{5}*{}       & 5.6291  & 3.0528  & 43.8212  &\multirow{5}*{}\\
  \hline
\end{tabular}
\begin{tablenotes}
\item[1] All variables are uniform units.
\item[2] Transferred energy is abbreviated as TE.
\end{tablenotes}
\end{center}
\end{threeparttable}
\end{table*}

To better evaluate the optimization problem, a \emph{data collection round} \cite{imon2015energy} is defined for a process where the sink collects sensing data from all sensor nodes, the sensing data is in turn transferred from leaf sensor nodes to sink over parent sensor nodes. In particular, the parent sensor nodes not only transmit received sensing data of child sensor nodes, but also transmit their own sensing data to their parent sensor nodes. In Fig. \ref{routing model}, a \emph{data collection round} is divided into 6 time slots according to the \emph{half-duplex} communication mode.
Using the same parameter settings, we perform the optimization problem under both \emph{orthogonal channel} (OC) and \emph{interference channel} (IFC) with no energy transfer and energy transfer, respectively. We attain the total network delay over time as shows in Fig. \ref{accumulateddelayround}.
\begin{figure}[htb]
\begin{center}
$\begin{array}{l}
\includegraphics[width=3.0in]{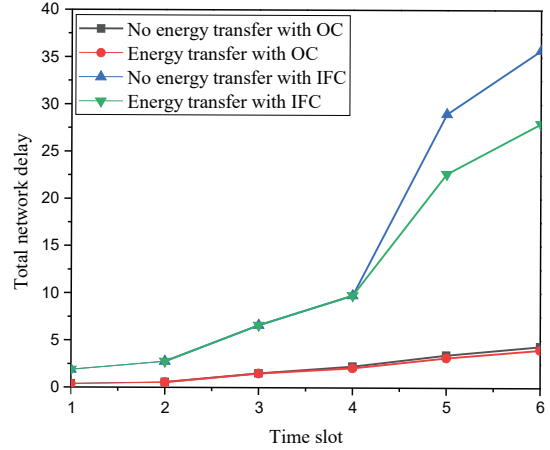}\\
 \end{array}$
\end{center}
\caption{Total delay of energy harvesting WSNs over time.} \label{accumulateddelayround}
\end{figure}
\subsection{Performance Analysis}
From Table \ref{Solution results} and Fig. \ref{accumulateddelayround}, we observe that some interesting results:
\begin{enumerate}
\item The network delay in the orthogonal channel is less than that in the interference channel. It means that the interference signals among data links significantly affect the total network delay in energy harvesting WSNs, which should not be ignored in the WSNs design.
\item In the models of orthogonal channel and interference channel, the network delay with no energy transfer is more than that with energy transfer. Since energy transfer between the energy-rich sensor nodes and the energy-hungry sensor nodes can help to decrease the total delay and enhance the total performance in WSNs.
\item In tree-based WSNs topologies, the sensor node is closer to the sink, the more energy is needed since it has heavier traffic loads. The total network delay also increases for the fixed channel gain.
\item In  our model, the power allocation of each active link is proportional to the amount of data flow and SINR.
\end{enumerate}

\section{Conclusion and Discussion} \label{conclusion}
We have investigated the optimal data rates, power allocations and energy transfers for minimizing the total delay in the energy harvesting WSNs with interference channel in a time slot. We have formulated the optimization problem which subjects to information rate requirements, energy and power consumption as a non-convex optimization problem under two cases, i.e., no energy transfer and energy transfer. By exploiting the convex approximation with  relatively high SINR, the optimization problem has been converted into a tractable convex problem. Moreover, we also have derived the properties of the optimal solution by Lagrange duality. Finaly, we solved the optimization problem by CVX solver. The experimental results shown that when data flow and energy topologies were fixed, the interference signals significantly effect the network performance; the energy transfer can help to decrease the total network delay; and the power allocation on each data link was proportional to the amount of data flow and SINR for the energy harvesting WSNs in a time slot.

Furthermore, our work can be further extended in some aspects. First, the approximate method only suits for the case of relatively high SINR and can not be used to deal with the case of low SINR in the network. Second, we can not provide a close-form solution for the optimization problem and only employed the experimental results to explain the optimization problem, making it difficult to carry out theoretical analysis on the relationship between data flow and energy flow under interference channel in a time slot. Moreover, the network topology of our work can be replaced by the others.  In the future, we will consider the above aspects.

\section*{Acknowledgment}
This work was supported by National Natural Science Foundation of China (No.61573277), the Open Projects Program of National Laboratory of Pattern Recognition and The Hong Kong Polytechnic University Research Committee for financial and technical support. The authors would like to thank Ms. L. Yue and Mr. Z. Wu from Xiamen University, and Mr. S. Shi and Dr. K. Shang from Southern University of Science and Technology for kind help and valuable discussions. We also thank the anonymous referees for their insightful comments and helpful suggestions which significantly improve the manuscript's quality.

\ifCLASSOPTIONcaptionsoff
  \newpage
\fi



%
\bibliographystyle{IEEEtran}

\bibliography{Reference}

\begin{thebibliography}{10}
\providecommand{\url}[1]{#1}
\csname url@samestyle\endcsname
\providecommand{\newblock}{\relax}
\providecommand{\bibinfo}[2]{#2}
\providecommand{\BIBentrySTDinterwordspacing}{\spaceskip=0pt\relax}
\providecommand{\BIBentryALTinterwordstretchfactor}{4}
\providecommand{\BIBentryALTinterwordspacing}{\spaceskip=\fontdimen2\font plus
\BIBentryALTinterwordstretchfactor\fontdimen3\font minus
  \fontdimen4\font\relax}
\providecommand{\BIBforeignlanguage}[2]{{%
\expandafter\ifx\csname l@#1\endcsname\relax
\typeout{** WARNING: IEEEtran.bst: No hyphenation pattern has been}%
\typeout{** loaded for the language `#1'. Using the pattern for}%
\typeout{** the default language instead.}%
\else
\language=\csname l@#1\endcsname
\fi
#2}}
\providecommand{\BIBdecl}{\relax}
\BIBdecl

\bibitem{ulukus2015energy}
S.~Ulukus, A.~Yener, E.~Erkip, O.~Simeone, M.~Zorzi, P.~Grover, and K.~Huang,
  ``Energy harvesting wireless communications: A review of recent advances,''
  \emph{IEEE Journal on Selected Areas in Communications}, vol.~33, no.~3, pp.
  360--381, 2015.

\bibitem{8269105}
S.~Leng, A.~M. Ibrahim, and A.~Yener, ``Energy cooperative multiple access
  channels with energy harvesting transmitters and receiver,'' in \emph{2017
  IEEE Globecom Workshops (GC Wkshps)}, Dec 2017, pp. 1--6.

\bibitem{bertsekas1992data}
D.~P. Bertsekas, R.~G. Gallager, and P.~Humblet, \emph{Data networks}.\hskip
  1em plus 0.5em minus 0.4em\relax Prentice-Hall International New Jersey,
  1992, vol.~2.

\bibitem{gurakan2016optimal}
B.~Gurakan, O.~Ozel, and S.~Ulukus, ``Optimal energy and data routing in
  networks with energy cooperation,'' \emph{IEEE Transactions on Wireless
  Communications}, vol.~15, no.~2, pp. 857--870, 2016.

\bibitem{boyd2004convex}
S.~Boyd and L.~Vandenberghe, \emph{Convex optimization}.\hskip 1em plus 0.5em
  minus 0.4em\relax Cambridge university press, 2004.

\bibitem{grant2008cvx}
M.~Grant and S.~Boyd, ``{CVX}: Matlab software for disciplined convex
  programming, version 2.1,'' \url{http://cvxr.com/cvx}, Mar. 2014.

\bibitem{xiao2004simultaneous}
L.~Xiao, M.~Johansson, and S.~P. Boyd, ``Simultaneous routing and resource
  allocation via dual decomposition,'' \emph{IEEE Transactions on
  Communications}, vol.~52, no.~7, pp. 1136--1144, 2004.

\bibitem{xi2008node}
Y.~Xi and E.~M. Yeh, ``Node-based optimal power control, routing, and
  congestion control in wireless networks,'' \emph{IEEE Transactions on
  Information Theory}, vol.~54, no.~9, pp. 4081--4106, 2008.

\bibitem{fouladgar2013information}
A.~M. Fouladgar and O.~Simeone, ``Information and energy flows in graphical
  networks with energy transfer and reuse,'' \emph{IEEE Wireless Communications
  Letters}, vol.~2, no.~4, pp. 371--374, 2013.

\bibitem{xu2017optimization}
W.~Xu, W.~Cheng, Y.~Zhang, Q.~Shi, and X.~Wang, ``On the optimization model for
  multi-hop information transmission and energy transfer in tdma-based wireless
  sensor networks,'' \emph{IEEE Communications Letters}, vol.~21, no.~5, pp.
  1095--1098, 2017.

\bibitem{imon2015energy}
S.~K.~A. Imon, A.~Khan, M.~D. Francesco, and S.~K. Das, ``Energy-efficient
  randomized switching for maximizing lifetime in tree-based wireless sensor
  networks,'' \emph{IEEE/ACM Transactions on Networking}, vol.~23, no.~5, pp.
  1401--1415, 2015.

\bibitem{huang2013spatial}
K.~Huang, ``Spatial throughput of mobile ad hoc networks powered by energy
  harvesting,'' \emph{IEEE Transactions on Information Theory}, vol.~59,
  no.~11, pp. 7597--7612, 2013.

\bibitem{adu2018energy}
K.~S. Adu-Manu, N.~Adam, C.~Tapparello, H.~Ayatollahi, and W.~Heinzelman,
  ``Energy-harvesting wireless sensor networks (eh-wsns): A review,'' \emph{ACM
  Transactions on Sensor Networks (TOSN)}, vol.~14, no.~2, p.~10, 2018.

\bibitem{gurakan2013energy}
B.~Gurakan, O.~Ozel, J.~Yang, and S.~Ulukus, ``Energy cooperation in energy
  harvesting communications,'' \emph{IEEE Transactions on Communications},
  vol.~61, no.~12, pp. 4884--4898, 2013.

\bibitem{fu2010fast}
L.~Fu, S.~C. Liew, and J.~Huang, ``Fast algorithms for joint power control and
  scheduling in wireless networks,'' \emph{IEEE Transactions on Wireless
  Communications}, vol.~9, no.~3, pp. 1186--1197, 2010.

\bibitem{johansson2003simultaneous}
M.~Johansson, L.~Xiao, and S.~Boyd, ``Simultaneous routing and power allocation
  in cdma wireless data networks,'' in \emph{Communications, 2003. ICC'03. IEEE
  International Conference on}, vol.~1.\hskip 1em plus 0.5em minus 0.4em\relax
  IEEE, 2003, pp. 51--55.

\bibitem{watfa2011multi}
M.~K. Watfa, H.~AlHassanieh, and S.~Selman, ``Multi-hop wireless energy
  transfer in wsns,'' \emph{IEEE communications letters}, vol.~15, no.~12, pp.
  1275--1277, 2011.

\end{thebibliography}

%





\end{document}